\documentclass[conference,a4paper,10pt]{IEEEtran}
\IEEEoverridecommandlockouts
\usepackage{cite}
\usepackage{amsmath}
\usepackage{amssymb,amsfonts,dsfont,amsthm}
\usepackage{algorithmic}
\usepackage{graphicx}
\usepackage{textcomp}
\usepackage{xcolor}
\usepackage{amsmath,amsfonts,graphicx,amssymb}
\usepackage{amsthm}
\usepackage{ifthen}
\usepackage{color}
\usepackage{caption}
\usepackage{subcaption}
\usepackage{mathtools}
\usepackage{bbm}
\usepackage{multirow}
\usepackage{footmisc}

\newtheorem{conjecture}{Conjecture}

\newtheorem{lemma}{Lemma}

\newtheorem{theorem}{Theorem}
\newtheorem{remark}{Remark}
\newtheorem{definition}{Definition}
\def\BibTeX{{\rm B\kern-.05em{\sc i\kern-.025em b}\kern-.08em
    T\kern-.1667em\lower.7ex\hbox{E}\kern-.125emX}}
\begin{document}

\title{Partial Server Pooling in Redundancy Systems 
\thanks{This work was supported in part by the Bharti Centre for Communications in IIT Bombay,
	grants from CEFIPRA, DST under the Indo-Korea joint programme of
	cooperation in Science and Technology.}
}

\author{
\IEEEauthorblockN{Akshay Mete, D. Manjunath, Jayakrishnan Nair}
\IEEEauthorblockA{\textit{Department of Electrical Engineering, IIT Bombay}}
\and
\IEEEauthorblockN{Balakrishna Prabhu}
\IEEEauthorblockA{\textit{LAAS-CNRS, Universit\'e de Toulouse}}
}
\newboolean{showcomments}
\setboolean{showcomments}{false}
\newcommand{\am}[1]{  \ifthenelse{\boolean{showcomments}}
	{ \textcolor{green}{(Akshay says:  #1)}} {}  }
\newcommand{\jk}[1]{  \ifthenelse{\boolean{showcomments}}
	{ \textcolor{blue}{(JK says:  #1)}} {}  }
\newcommand{\dm}[1]{  \ifthenelse{\boolean{showcomments}}
	{ \textcolor{red}{(DM says:  #1)}} {}  }
\newcommand{\bala}[1]{  \ifthenelse{\boolean{showcomments}}
	{ \textcolor{green}{(Bala says:  #1)}} {}  }

\newboolean{itc}
\setboolean{itc}{false}
\newcommand{\conftext}[1]{\ifthenelse{\boolean{itc}}{#1} {}  }
\newcommand{\TRtext}[1]{\ifthenelse{\boolean{itc}} {} {#1} }

\newcommand{\ra}{\rightarrow} 
\newcommand{\ua}{\uparrow}
\newcommand{\da}{\downarrow}
\newcommand{\prob}[1]{P\left(#1\right)} 
\newcommand{\imp}{\Rightarrow}
\newcommand{\R}{\mathbb{R}} 
\newcommand{\Exp}[1]{\mathbb{E}\left[#1\right]} 
\newcommand{\eqdist}{\stackrel{d}{=}}
\newcommand{\leqsim}{\lesssim}
\newcommand{\leqst}{\leq_{\mathrm{st}}}
\newcommand{\leqas}{\leq_{\mathrm{a.s.}}}
\newcommand{\indicator}{\boldsymbol{1}}
\newcommand{\ceil}[1]{\left\lceil #1 \right\rceil}
\newcommand{\floor}[1]{\left\lfloor #1 \right\rfloor}
\newcommand{\ignore}[1]{}
\newcommand{\norm}[1]{\|#1\|}
\newcommand{\fr}[1]{\left\{#1\right\}}
\newcommand{\N}{\mathbb{N}} 
\newcommand{\M}{\mathcal{M}}
\newcommand{\coc}{\mathsf{c.o.c.}}
\renewcommand{\cos}{\mathsf{c.o.s.}}
\allowdisplaybreaks[4]

\maketitle

\begin{abstract}

 Partial sharing allows providers to possibly pool a fraction of 
their resources when full pooling is not beneficial to them.
Recent work in systems without sharing has shown that redundancy 
can improve performance considerably. In this paper, we combine partial sharing
and redundancy by developing partial sharing models for providers operating 
multi-server systems with redundancy. Two M/M/N queues with redundant service models are considered. 
Copies of an arriving job are placed
in the queues of servers that can serve the job. Partial sharing
models for cancel-on-complete and cancel-on-start redundancy models
are developed. For cancel-on-complete, it is shown that the Pareto
efficient region is the full pooling configuration. For a
cancel-on-start policy, we conjecture that the Pareto frontier is
always non-empty and is such that at least one of the two providers
is sharing all of its resources. For this system, using bargaining
theory the sharing configuration that the providers may use is
determined.  Mean response time and probability of waiting are the
performance metrics considered.
\end{abstract}


\maketitle
\section{Introduction}
\label{section:intro}
We consider resource sharing by service systems modeled as multi
server queueing systems, e.g., server farms, cloud computing systems,
call centers, inventory systems, and emergency services. These
services dimension their resources (e.g., number of servers) to
provide a prescribed quality of service (QoS). Two commonly used QoS
measures in delay systems (as opposed to loss systems) are the
probability of waiting for service (famously characterized by the
Erlang-C formula for M/M/N queues), and waiting and/or sojourn time
moments.

For many of the above mentioned systems, resources are expensive and
different independent systems could possibly share their resources to
improve their customers' QoS. It may also be that procurement of
additional resources takes time, and sharing could be a useful interim
measure. In \cite{Karsten11b}, two models for resource sharing among
different service providers are identified.~(1)~Providers pool their
existing resources with the expectation that the joint system is
beneficial over operating alone. (2)~Providers jointly determine the
total resources for the QoS requirements of the combined system. For
both systems, cooperative game theory is used to determine the cost
shares among the coalition of providers.

Our interest in this paper is related to the first kind of system
above, but in the setting of non-transferable utility. In this case,
the providers may not always have the incentive to completely pool
their resources, as the following example illustrates. Consider two
service providers, 1 and 2, modeled as $M/M/N$ queues. The providers
have, respectively, 20 and 30 servers and an offered load of,
respectively, 16 and 28 Erlangs. When operating alone, the providers'
QoS, measured as Erlang-C probabilities, are, respectively, 0.25 and
0.62. When the providers merge to create a coalition system of 50
servers with 44 Erlangs load, the QoS in the joint system is
0.28. Clearly, the first provider is not incentivized to join the
`naive' full pooling coalition. 
A natural question then is to seek partial pooling models that may
incentivize both providers to join the coalition. Here by partial pooling models we mean that each provider contributes a fraction (could be all) of its resources into a common pool which can then be used to serve requests from any provider.

The focus of this work is to develop partial sharing models in delay
systems and answer two key questions:~How to share?~How much to share?
We consider delay systems where arriving jobs are replicated into
queues at the servers that can serve them.  Sending redundant copies
of a job to queues at servers that can service them is of interest
due to their use in several systems like call centers and cloud
server farms. Two redundancy models are
popular---\emph{cancel-on-complete} and \emph{cancel-on-start}.
Analytical studies of these models without partial sharing
are available, among others,
\cite{gardner2017redundancy,bonald2017performance,Ayesta2018}.



The rest of the paper is organized as follows: In the next section, we
describe our system model involving two providers operating
multi-server service systems, and present our redundancy-based partial
sharing mechanism.
In Section \ref{section:coc}, we analyze partial sharing for
cancel-on-complete systems. Exploiting recent results
of \cite{bonald2017performance}, we show that full sharing is the only
Pareto-optimal configuration.
In Section \ref{section:cos}, we consider partial sharing in the
cancel-on-start system. We obtain the stationary distributions of the
number in the system, and for a special case, we show that the Pareto
region is such that at least one of the providers shares all of its
resources. Based on numerical evidence, we conjecture that this is
true in general. We then use bargaining theory to capture the stable
sharing agreement. We conclude in Section~\ref{section:discussions} by
showing that the results of cancel-on-complete are directly applicable
to the joint system that uses a single server (capacity equal to sum
of server capacities of the two systems) with two queues served
according to a balanced fair rate allocation. We also discuss related
literature on resource pooling and future work.

\section{System Model}
\label{section:system_model}

Consider two service providers $P_1$ and $P_2$ with $N_1$ and $N_2$
servers, respectively. The servers are homogenous with unit service
rate. Jobs of provider $P_i$ arrive according to a Poisson process of
rate $\lambda_i$, and the service requirements (a.k.a. sizes) of jobs
are i.i.d. exponential with mean $1/\nu_i < \infty.$ $\rho_i
:= \lambda_i/\nu_i$ denotes the traffic load 
$P_i.$ 

Each server has its own queue, and serves jobs using FCFS
discipline. Both providers use redundancy as follows. When a job
arrives, $d$ copies (a.k.a. replicas) of this job are sent to $d$
different servers. Further, both providers replicate copies to all
servers that can process their jobs. In the {\em
	standalone} \footnote{`Standalone' refers to the system with no pooling
	between providers.} system, this implies that $d = N_i$ for provider $i.$ On
the other hand, in a pooled system, $d_i$ can be larger than $N_i$ and
depends on the number of servers shared by the other provider
$-i.$\footnote{While referring to provider $i$, we use ${-i}$ to refer
	to the other provider.}

\bala{Can we say that the results also hold for any other $d$? And, that for
	simplicity we assume $d=N_i$?}
\jk{I don't have any understanding of how this would work with arbitrary $d_i.$
	Good question to think about in the future! But I'd be wary of making
	any such claim here.}

Two types of redundancy models are commonly studied in the
literature. The redundant copies of a job can either be removed from
the system at the instance when first of its copies starts service,
i.e., \emph{cancel-on-start} ($\cos$) replication, or at the instance
when first of its copies finishes its service, i.e., \emph{cancel-on-complete} 
($\coc$) replication. We analyze both $\cos$ as
well as $\coc$ policies in the next two sections.

For stability, we assume $\rho_i < N_i$ for $i = 1,2.$ This condition
is necessary and sufficient for both $\coc$ and $\cos$ when the
arrival process is Poisson and replica sizes are i.i.d. with an
exponential distribution (see \cite{bonald2017performance} for $\coc$
and \cite{Ayesta2018} for $\cos$).

We shall consider two different performance metrics for the service
providers: {\em (i)} the {\em stationary waiting probability}, defined
as the steady state probability that an arriving job has to wait for
service; and {\em (ii)} the {\em stationary mean response time.} We
make the following important remark at this point.

\begin{remark}
	Since each server has its own waiting line, strictly speaking, the
	{\em standalone} systems are not Erlang-C systems. However, the $\cos$
	system in which copies are replicated to all the servers is indeed
	equivalent to an Erlang-C system \cite{Ayesta2018}. It thus makes
	sense to refer to the waiting probability in the standalone $\cos$
	system as the Erlang-C probability.
\end{remark}

In light of the above remark, the performance metrics of the Erlang-C
system will serve as another benchmark when highlighting the benefits
of partial sharing compared to the no sharing case.
For provider $P_i,$ the standalone Erlang-C probability
$C_i^{\text{s}}$ and the stationary mean response time
$D_i^{\text{s}}$ are given by
$$C_i^{\text{s}} =\frac{\rho_i^{N_i}}{\rho_i^{N_i}+{\big(1-\frac{\rho_i}{N_i}\big) \left(\sum\limits_{k=0}^{N_i-1}\frac{\rho_i^k}{k!}\right)N_i!}},\ 
D_i^{\text{s}}  =\frac{1}{\nu_i}+\frac{C_i^{\text{s}}}{\nu_i(N_i-\rho_i)}.$$

\subsection{Partial Sharing Policy}
We propose a partial resource sharing policy, where each of the
providers contributes some of its servers to a common pool. The
servers in this common pool can serve the jobs from both of the
service providers. Hence the system has three types of servers
depending on the types of jobs they can serve. We now formally
define the partial sharing policy.

The partial sharing policy is parametrized by $(k_1,k_2),$ where
$k_i \in \{0,1,2,\dots,N_i\}$ is the number of servers contributed by
provider $P_i$ to the common pool. Hence these $N_1+N_2$ servers are
classified in the following three separate pools.
\begin{itemize}
	\item  Dedicated servers of provider $P_1$: $N_1-k_1$ dedicated servers which can serve only jobs of provider $P_1$.
	\item  Dedicated servers of provider $P_2$: $N_2-k_2$ dedicated servers which can serve only jobs of provider $P_2.$
	\item Common pool: $k_1+k_2$ shared servers which can serve jobs from both providers $P_1$ and $P_2$. 
\end{itemize} 
\TRtext{
	\begin{figure}
		\centering \includegraphics[scale=.8]{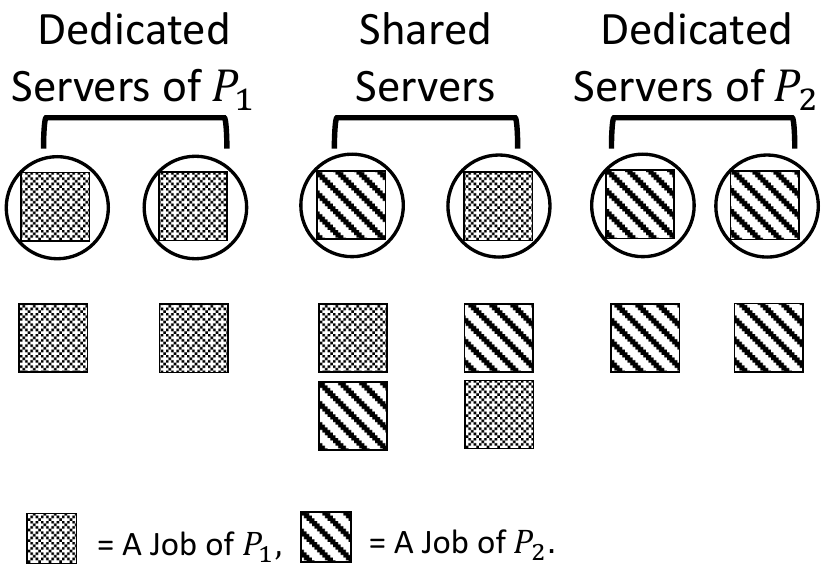} \caption{Partial
			sharing policy}
		\label{fig:bp-rateregion}
	\end{figure}
}
On arrival of a provider $i$ job into the system, copies of this job
are sent to all the $N_{i}+k_{-i}$ servers that can serve it, i.e., to
the $N_{i}-k_{i}$ dedicated servers of provider $P_i$ and
$k_{1}+k_{2}$ servers in the shared pool.

%

{\em Notations:} For a partial sharing configuration $(k_1,k_2)$, for
provider $P_i$, the waiting probability is denoted by $C_i(k_1,k_2)$
and the mean response time will be denoted by $D_i(k_1,k_2)$. To keep
the notation simple, we shall use the same notation for the two
performance metrics in both the $\coc$ as well as the $\cos$
systems. Since these two systems are treated in separate sections, no
confusion should arise.

\subsection{Pareto-frontier}

Each provider is assumed to be optimizing its own performance metrics,
i.e., for a provider to consider sharing its servers, it has to
benefit from doing so. To determine which partial sharing policy will
be acceptable to the providers, we use the concept of {\em
	Pareto-frontier} which is widely used in economics and multi-objective
optimization.

Let $B_i$ be either $C_i$ or $D_i$.  A policy/configuration $(k_1,
k_2)$ is said to be {\em Pareto-optimal}, if,
\begin{enumerate}
	\item $B_i(k_1,k_2) < B_i(0,0)$ for $i = 1,2,$ and 
	\item there does not exist another policy $(l_1, l_2)$ such that
	$B_i(l_1, l_2) \leq B_i(k_1, k_2),$ for $ i=1,2,$ with strict
	inequality for at least one $i$.
\end{enumerate}

That is, a Pareto-optimal configuration is one that results in
improved performance for each provider, and for which there does not
exist any other policy that is better for both the
providers.\footnote{The first condition is not part of the standard
	definition of Pareto optimality. However, in the present context,
	since Pareto-optimal configurations are meant to capture possible
	agreement points between the providers, it is natural to impose this
	condition of individual rationality.} The Pareto-frontier,
$\mathcal{P}$ is defined as the set of all Pareto optimal policies and
is the set of policies for which both providers benefit individually
compared to policies outside this set. Within the set, two policies
are not comparable since one provider gains while the other
loses. Existence of a non-empty $\mathcal{P}$ implies that partial
sharing can benefit both providers when compared to not sharing. In
the next two sections, we show that $\mathcal{P}$ is indeed non-empty
for both $\coc$ and $\cos$ systems.

\section{Partial pooling via cancel-on-complete replication}
\label{section:coc}

In this section, we explore partial pooling via the resource sharing
mechanism described in Section~\ref{section:system_model} with
\emph{cancel-on-complete} ($\coc$) replication.
Specifically, each incoming job of provider~$i$ releases a replica of
the job to all eligible servers (the $N_i - k_i$ servers in the
dedicated pool of $P_i$, and the $k_1 + k_2$ servers in the common
pool) upon arrival. The sizes of these replicas are assumed to be
i.i.d. and exponentially distributed with mean~$1/\nu_i$ (we comment
on this assumption later). The job gets completed when the first of
its replicas completes service, at which point the remaining replicas
are cancelled.

Our main result is that under $\coc$ replication, complete pooling
(i.e., $k_i = N_i\ \forall\ i$) is the \emph{only} Pareto-optimal
partial sharing configuration between the providers for the mean
response time metric. In other words, full pooling is not just optimal
for overall system performance, but also \emph{individually optimal}
from the standpoint of each provider. Thus, complete pooling is the
only reasonable configuration that the providers would agree upon in a
bargaining setting. In contrast, under \emph{cancel-on-start}
replication (discussed in Section~\ref{section:cos}), the
Pareto-frontier is a continuum of partial sharing configurations,
which may not include the complete pooling configuration.

In the following, we first obtain, by invoking recent results by
Bonald et al. (see \cite{bonald2017performance}), an expression for
the steady state mean response time under our partial pooling model
with $\coc$ replication. This enables us to characterize certain
monotonicity properties of the mean response time in the sharing
parameters. Finally, using these monotonicity properties, we determine
the Pareto-frontier of sharing configurations.

\begin{table*}[th]
	\centering
	\caption{Comparison of mean response time for the case
		$N_1 = N_2 = N,$ the standalone waiting probabilities of $S_1$
		and $S_2$ without replication being 5\% and 10\%, respectively and $\nu_1=\nu_2=1$.}
	\label{tb:bo_same_size}
	
	\begin{tabular}{|l|c|c|c|c|c|c|}
		\hline
		$N$  &\multicolumn{2}{|c|}{$D_1$}&\multicolumn{2}{|c|}{$D_2$}&\multicolumn{2}{|c|}{$D_1=D_2=D$}\\
		\hline
		&  Standalone system& Standalone system &  Standalone system& Standalone system&  Full Sharing & Naive 
		\\
		&   with \emph{c.o.c.} replication & without replication
		&with \emph{c.o.c.} replication & without replication &   with \emph{c.o.c.} & Full Sharing 
		\\\hline
		
		5  & 0.3231&1.0161 & 0.3722&1.0372 & 0.1730&1.0220  \\
		10 & 0.2121&1.0106 & 0.2491&1.0249 & 0.1145& 1.0015 \\
		15 & 0.1679&1.0084 & 0.1988&1.0199&  0.0910&   1.0012 \\
		20 & 0.1428&1.0071 & 0.1699&1.0170&  0.0776& 1.0011   \\
		\hline
	\end{tabular}
\end{table*}

\subsection{Performance characterization under $\coc$ replication}

In \cite{bonald2017performance}, Bonald et al. establish an
equivalence between the distribution of the steady state system
occupancy vector in a multiclass $\coc$ queueing system, and that in a
single server system with balanced fair scheduling. Specifically, for
a given partial sharing configuration $(k_1,k_2),$ consider the
following two systems.

\noindent $\boldsymbol{\mathcal{S}1}$:
$\mathcal{S}1$ is a multiclass queueing system with two classes
(corresponding to the two providers) and $N_1 + N_2$ servers. The
servers are identical and have a unit service rate. Jobs of
class~$i$ (i.e., corresponding to Provider~$i$) can be served on the
$N_i-k_i$ dedicated servers of provider~$i$ as well as on the
$k_1+k_2$ servers in the common pool. An incoming job is replicated
on all eligible servers in $\coc$ mode.

\noindent $\boldsymbol{\mathcal{S}2}$:
$\mathcal{S}2$ is a two-class single server system. The two job
classes correspond to the two providers, and the system maintains a
separate queue for the active jobs of either class. The server has a
service rate of $N_1 + N_2.$ Let $n_i$ denote the number of jobs of
class~$i$ in the system. For a given system state~$(n_1,n_2),$ each
class $i$ is allotted a service rate $r_i,$ where $(r_1,r_2)$ is the
balanced fair rate allocation corresponding to the polymatroidal
rate region $\mathcal{R}(k_1,k_2)$ defined as follows:
\begin{align*}
\mathcal{R}(k_1,k_2) &= \{(r_1,r_2) \in \R_+^2:\ r_1\leq N_1+k_2,\; \\
& \qquad r_2\leq N_2+k_1,\; r_1+r_2 \leq N_1+N_2 \}.
\end{align*}
\TRtext{The above rate region is also depicted in
	Figure~\ref{fig:bf-rateregion}.} We refer the reader
to \cite{bonald2003insensitive} for a detailed description of balanced
fair scheduling. Proposition~1 in~\cite{bonald2017performance} shows that the steady
state average system occupancies in $\mathcal{S}1$ and $\mathcal{S}2$ coincide. From this
equivalence, and using known results for the steady state mean
response times under balanced fair scheduling
(see \cite{bonald2003insensitive}), one obtains the following
characterization of the mean response time for our partial pooling
model under $\coc$ replication.
\TRtext{
	\begin{figure}[t]
		\centering
		\includegraphics[width = 0.3\textwidth]{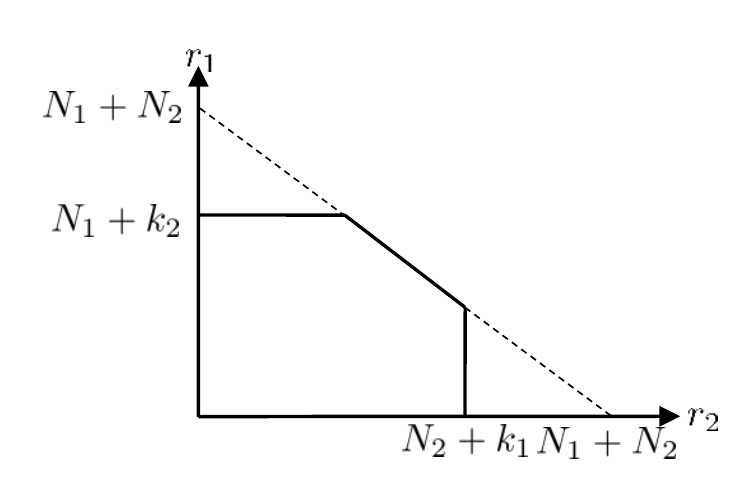}
		\caption{Rate region $\mathcal{R}(k_1,k_2)$ for BF-based partial pooling}
		\label{fig:bf-rateregion}
	\end{figure}
}
\ignore{
	We invoke the result from (,), which states that the state space
	distribution of a system under partial sharing configuration
	$(k_1,k_2)$ with \emph{c.o.c.} is same as that of a balanced fair
	scheduling policy \cite{bonald2003insensitive} described as
	follows:\\\\
	\textit{Balanced Fair Scheduling}: Let us consider the case where the
	servers of both providers can be combined into a single server with
	service rate $N_1+N_2$. Moreover, the service capacity of this merged
	server is arbitrarily (and dynamically) divisible between the two
	providers. We formulate a partial pooling model for this case based on
	balanced fair scheduling.
	
	The BF-based partial pooling mechanism is parametrized by $(k_1,k_2),$
	where $k_i \in [0, N_i]$ is a measure of the extent to which $S_i$ is
	willing to share its service capacity with $S_{-i}.$ Specifically, let
	$n_i$ denote the number of unfinished jobs of $S_i$ in the
	system. Given the system state $(n_1,n_2),$ the BF-based partial
	pooling mechanism awards a service rate $r_i$ to Provider~$i,$ such
	that $(r_1,r_2)$ is the balanced fair allocation over the
	polymatroidal rate region $\mathcal{R}(k_1,k_2)$ defined as follows:
	\begin{align*}
	\mathcal{R}(k_1,k_2) = \{(r_1,r_2) \in &\R_+^2:\ r_1\leq N_1+k_2,\;
	r_2\leq N_2+k_1,\;\\
	& \quad r_1+r_2 \leq N_1+N_2 \}.
	\end{align*}
	Further, within each queue $i,$ the allocated service rate $r_i$ is
	split equally among the waiting jobs, i.e., in a processor sharing
	manner.
	
	The rate region is also depicted in
	Figure~\ref{fig:bf-rateregion}. Note that $k_i$ can be interpreted as
	the maximum additional service rate made available to $S_{-i}$ by
	$S_i.$ Thus, a larger value of $k_i$ implies a greater willingness on
	the part of $S_i$ to share its service capacity (dynamically) with
	$S_{-i}.$
	
	We now give a brief description on how the balanced fair rate
	allocation is determined over the rate region $\mathcal{R}(k_1,k_2)$;
	the interested reader is referred to \cite{bonald2003insensitive} for
	more details. The rate allocation is defined in terms of the balance
	function. $\Phi(0,0)=1$ and $\Phi(n_1,n_2),$ which is in turn defined
	recursively as follows:
	\begin{align}\label{eq:balance function}
	\Phi(n_1,n_2)\nonumber
	=\begin{cases}
	\frac{\Phi(n_1-1,n_2)}{N_1+k_2}, \text{ if } n_1 > 0 \text{ or }	n_2 =0\\
	\frac{\Phi(n_1,n_2-1)}{N_2+k_1}, \text{ if } n_1=0 \text{ and } n_2>0\\
	\frac{\Phi(n_1-1,n_2)+\Phi(n_1,n_2-1)}{N_1+N_2}, \text{ otherwise.}
	\end{cases}
	\end{align}
	The rate allocation is then given by
	$$r_1(n_1,n_2)=\frac{\Phi(n_1-1,n_2)}{\Phi(n_1,n_2)}, \;\;r_2(n_1,n_2)=\frac{\Phi(n_1,n_2-1)}{\Phi(n_1,n_2)}.$$It can be shown
	that the above allocation is efficient (i.e., it results in an
	allocation on the outer boundary of the rate region), and induces a
	product-form stationary distribution on the queue occupancy vector
	$(n_1,n_2)$ that is insensitive to the job size distributions beyond
	their means \cite{bonald2003insensitive}.
}

\ignore{
	
	In Lemma \ref{lemma:coc_erlangc}, we provide the expression for
	Erlang-C probability under partial configuration $(k_1,k_2)$
	with \emph{c.o.c.}
	\begin{lemma}\label{lemma:coc_erlangc}
		For a partial configuration $(k_1,k_2)$ with \emph{c.o.c.} policy, the waiting probability is given by:
		\begin{align*}
		C_i(k_1,k_2)&=1 -\frac{1}{G(k_1,k_2)}\frac{N_{-i}+k_i}{N_{-i}+k_i-\rho_{-i}}
		\end{align*}
		where,
		\begin{align*}
		G(k_1,k_2)=\frac{1}{1-\frac{\rho_1+\rho_2}{N_1+N_2}}\bigg(\frac{1-\frac{\rho_1}{N_1+k_2}}{1-\frac{\rho_1}{N_1+N_2}}+\frac{1-\frac{\rho_2}{N_2+k_1}}{1-\frac{\rho_2}{N_1+N_2}}-1\bigg)
		\end{align*}
	\end{lemma}
	\begin{proof}
		Please see Appendix \ref{appendix:coc_erlangc}
	\end{proof}
}

\begin{lemma}\label{lemma:coc_waiting}
	For partial sharing configuration $(k_1,k_2)$ with $\coc$
	replication, the steady state mean response time corresponding to jobs of provider~$i$, denoted $D_i(k_1,k_2)$ is given by
	$D_i(k_1,k_2)=\nu_i^{-1}\Bigg(\frac{1}{N_1+N_2-\rho_1-\rho_2}+\bigg(\frac{\left(1-\frac{\rho_i}{N_1+N_2}\right)\frac{1}{N_i+k_{-i}}}{\left(1-\frac{\rho_i}{N_i+k_{-i}}\right)^2}
	-\frac{\frac{1}{N_1+N_2}}{1-\frac{\rho_i}{N_1+k_2}}\bigg)\left(\frac{1-\frac{\rho_1}{N_1+N_2}}{1-\frac{\rho_1}{N_1+k_2}}+\frac{1-\frac{\rho_2}{N_1+N_2}}{1-\frac{\rho_2}{N_2+k_1}}-1\right)^{-1}\Bigg).$
\end{lemma}
\TRtext{The proof of Lemma \ref{lemma:coc_waiting} can be found in Appendix \ref{app:coc_waiting}.} \conftext{The proof of Lemma \ref{lemma:coc_waiting} can be found in \cite{Techrep}.} Given the above performance characterization, we now analyse the Pareto-frontier of partial sharing configurations.

\subsection{Pareto-optimal sharing configurations}

We first state the following monotonicity properties of the mean
response time under $\coc$ replication with respect to the sharing
paramters $k_1$ and $k_2.$

\begin{lemma} \label{lemma:coc_mono} Under $\coc$ replication,
	\begin{enumerate}
		\item $D_i(k_1,k_2)$ is a strictly decreasing function of $k_{-i},$
		\item $D_i(k_1,k_2)$ is a strictly increasing function of $k_{i}$ when $k_{-i} < N_{-i},$
		\item $D_i(k_1,k_2)$ is insensitive in $k_{i}$ when $k_{-i} = N_{-i}.$
	\end{enumerate}
\end{lemma}
\conftext{The proof of
	Lemma \ref{lemma:coc_mono} can be found in \cite{Techrep}.}\TRtext{The proof of Lemma~\ref{lemma:coc_mono} can be found in
	Appendix~\ref{appendix:coc_mono}.} The first
statement of the lemma states that $P_i$ benefits from additional
servers contributed by $P_{-i}$ to the common pool. Statement~(2)
and~(3) deal with the impact of $P_i$'s contribution to the common
pool on its own performance. Interestingly, this dependence depends on
the extent to $P_{-i}$'s contribution. When $P_{-i}$ does not
contribute all its servers to the common pool (i.e., $k_{-i} <
N_{-i}$), $P_i$'s own performance deteriorates as it contributes
additional servers to the common pool. However, if $P_{-i}$ has
contributes all its servers to the common pool (i.e., $k_{-i} =
N_{-i}$), then $P_i$'s performance is insensitive to its own
contribution to the common pool. This insensitivity plays a crucial
role in determining the Pareto-frontier of partial sharing
configurations.

\begin{theorem}\label{thm:coc_pareto}
	Under $\coc$ replication, complete pooling (i.e., $k_1=N_1,
	k_2=N_2$) is the only Pareto-optimal configuration.
\end{theorem}
\begin{proof}
	
	It follows from Lemma \ref{lemma:coc_mono} that
	\begin{align*}
	D_1(k_1,k_2) & \stackrel{(a)}\geq D_1(k_1,N_2) = D_1(N_1,N_2), \\
	D_2(k_1,k_2) & \stackrel{(b)}\geq D_2(N_1,k_2) = D_2(N_1,N_2).
	\end{align*}
	Moreover, the inequality $(a)$ (respectively, $(b)$) is strict when
	$k_1 < N_1$ (respectively, $k_2 < N_2$).
	\ignore{
		The first property of Lemma \ref{lemma:coc_mono} implies that
		$$D_1(k_1,N_2)<D_1(k_1,k_2) \text{ and }
		D_2(N_1,k_2)<D_2(k_1,k_2),$$ while the third property from
		Lemma \ref{lemma:coc_mono} implies that
		$$D_1(N_1,N_2)=D_1(k_1,N_2) \text{ and
		}D_2(N_1,N_2)=D_2(N_1,k_2) $$ }
	Therefore, $$(N_1,N_2)=\{\arg\min D_1(k_1,k_2)\}\cap\{ \arg\min
	D_2(k_1,k_2)\}$$ i.e., \emph{only} at $k_1=N_1,$
	$k_2=N_2$ are both $D_1(k_1,k_2)$ and $D_2(k_1,k_2)$
	minimized. Equivalently, all partial sharing configurations except the full
	sharing configuration will have a higher mean response time
	for at least one of the service providers. Therefore, the full
	sharing configuration i.e., $k_1=N_1$ and $k_2=N_2$ is the only
	Pareto-optimal configuration.
\end{proof}

From a bargaining standpoint, Theorem~\ref{thm:coc_pareto} states that
complete pooling is the only mutually agreeable configuration the
providers would settle upon under $\coc$ replication. This is in stark
contrast to the case of $\cos$ replication (addressed in the following
section), where complete pooling may not even be
Pareto-optimal. Interestingly, under $\coc$ replication, full pooling
is also a Nash equilibrium between the providers.

We conclude with a couple of remarks on $\coc$ replication in the
context of partial pooling.

\begin{remark} When neither operator contributes to the pool,
	i.e., $(k_1,k_2) = (0,0)$, the system is not equivalent to two
	standalone Erlang-C systems (one for each operator).  At $(0,0)$,
	each operator replicates copies to all of its servers. It is not
	hard to show that $D_i(0,0) < D_i^{s},$ where $D_i^{s}$ refers to
	the mean response time in the standalone Erlang-C system.
\end{remark}
\begin{remark} The stationary waiting
	probability is not a particularly meaningful metric under $\coc$
	replication. Indeed, $\coc$ replication is often used in applications
	where the beginning of service is hard to detect.
	Moreover, the power of the $\coc$ redundancy system depends to quite
	some extent on the assumptions of independent copies and exponential
	service time distributions. When copies are being served in parallel
	on several servers, these two assumptions imply that the service rate
	of the job is equal to the sum of the rates of the servers. It is as
	though the servers had pooled their service rates to serve the job.
	This can considerably increase the service rate for a job and hence
	decrease the mean response time. It could thus happen that under
	complete pooling, a job enters into service later than in the
	standalone Erlang-C system (or in a partially pooled system), but
	finishes service earlier. This could result in the stationary
	probability of wait being higher in the completely pooled system, even
	though the mean response time is the lowest possible.
\end{remark}

\subsection{Numerical study}

We now present the results from some numerical experiments to
illustrate the benefits of resource pooling via $\coc$ replication. In
Table~\ref{tb:bo_same_size}, we present the results for the case $N_1
= N_2 = N,$ with the arrival rates set so that the standalone
(Erlang-C) probability of wait for $P_1$ and $P_2$ are 0.05 and 0.1,
respectively. We take $\nu_1 = \nu_2 = 1.$ Note that with complete
pooling, which is the only Pareto-optimal sharing configuration,
$D_1(N_1,N_2) = D_2(N_1,N_2).$ We compare the mean response time under
this Pareto-optimal configuration with the the mean response time in
the standalone (Erlang-C) systems, the case $k_1 = k_2 = 0$
(standalone $\coc$ replicated systems), as well as naive full pooling
without replication (an Erlang-C system with $N_1 + N_2$ servers, and
arrival rate $\lambda_1 + \lambda_2$). Note that complete resource
pooling with $\coc$ replication results in a considerable reduction of
mean response time for both providers, sometimes even by an order of
magnitude.

\section{Partial pooling via cancel-on-start replication}
\label{section:cos}

In this section, we explore partial resource pooling
via \emph{cancel-on-start} ($\cos$) replication. In this model,
incoming jobs of provider $i$ get replicated at all the $N_i + k_{-i}$
eligible servers as before. However, as soon as the first replica
begins service, the others get cancelled. Thus, under $\cos$
replication, only one replica actually begins service, making it more
attractive in applications where parallel processing of replicas is
either infeasible (as in call centers) or undesirable (e.g., because
the replicas have comparable size, making parallel processing of
replicas inefficient). Throughout this section, we assume that $\nu_1
= \nu_2 = \nu.$

Partial pooling under $\cos$ replication is equivalent to a
hypothetical \emph{join-the-least-workload} system, where each server
maintains a FCFS queue, and an incoming job gets dispatched on arrival
to that eligible server that has the least unfinished work.  An
alternative and equivalent view of our $\cos$-based model is the
following: All jobs (from both providers) wait in a single FCFS queue,
and each server processes the earliest arriving eligible job. This
latter view makes our model an instance of the multi-server,
multi-class system analysed by Visschers et
al. in \cite{Visschers2012}, for which a product-form description of
the stationary distribution is available. An example of this
equivalence with three servers is shown in Fig.~\ref{fig:cosequiv}.
\begin{figure}[t!]
	\centering
	\includegraphics[width=0.45\textwidth]{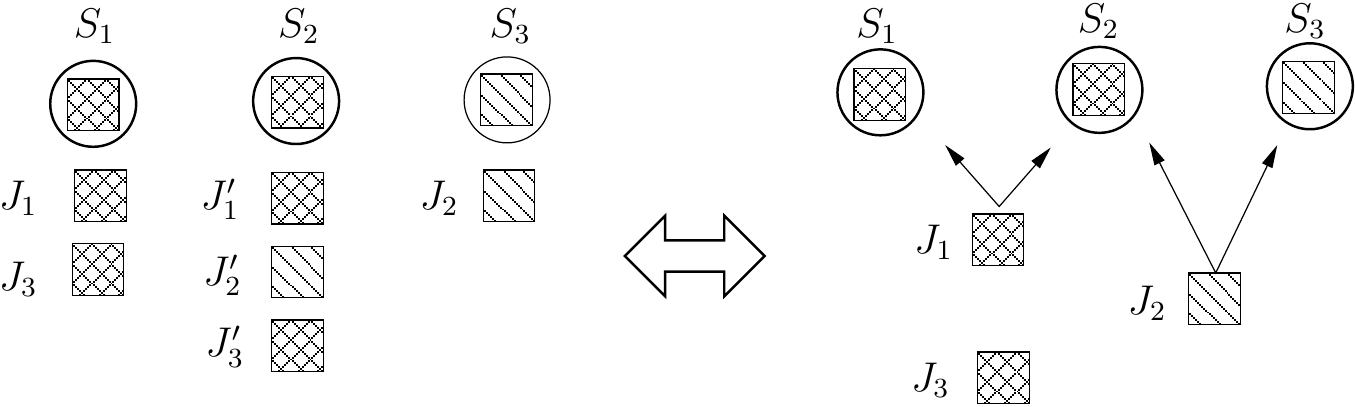}

	\caption{Equivalence of $\cos$ redundancy and a single FCFS queue. 
		The copy of a job is indicated by $'$. So, $J_1$ and $J_1'$ are copies of the same job.}
	\label{fig:cosequiv}
\end{figure}

The contributions of this section are as follows.

\noindent {\bf 1.} We use the framework
in \cite{Visschers2012} to obtain a characterization of the stationary
probability of waiting as well as the stationary mean response time
for each provider, under the $\cos$-based partial sharing model.

\noindent {\bf 2.} Since the expressions for the above performance metrics are
fairly involved, we are only able to analytically characterize the
Pareto-frontier for the stationary waiting probability metric when
$N_1 = N_2 = 1.$ In this case, by suitably extending the space of
sharing configurations to $[0,1]^2$ via time-sharing across
configurations in $\{0,1\}^2$, we show that Pareto-optimal sharing
configurations involve at least one provider \emph{always}
contributing its server to the common pool (i.e., $k_i = 1$ for
some~$i$). Intuitively, under efficient partial sharing
configurations, the more congested provider always places its server
in the common pool, whereas as the less congested provider places its
server in the common pool for some (long-run) fraction of time.

\noindent {\bf 3.} We conjecture that the above structure of the Pareto-frontier
holds for $N_1,N_2 \geq 1.$ This conjecture is validated via numerical
evaluations.

\noindent {\bf 4.}
Finally, we invoke the Kalai-Smorodinsky solution from bargaining
theory to capture the partial sharing configuration that the providers
would agree upon. Via numerical experiments, we demonstrate the
potential benefits from partial resource pooling between service
providers.

We begin with the following remark on the special cases of no pooling
and complete pooling.
\begin{remark}
	The case $k_1 = k_2 = 0$ corresponds to provider~$i$ performing $\cos$
	replication among its own $N_i$ servers, which is in turn equivalent
	to an $M/M/N_i$ (Erlang-C) system. Thus, $$C_i(0,0) =
	C^s(\rho_i,N_i) \qquad (i = 1,2).$$ Complete pooling corresponds to an
	aggregated Erlang-C, which implies that $$C_1(N_1,N_2) = C_2(N_1,N_2)
	= C^s(\rho_1 + \rho_2,N_1 + N_2).$$
\end{remark}
\jk{Erlang-C notation should be consistent with
	Section~2.}

Next, we formulate our partial pooling model with $\cos$ replication
in the framework of \cite{Visschers2012}.

\subsection{Performance characterization using the framework of \cite{Visschers2012}}

Let $\mathcal{M} = \{1,2,\cdots, N_1 + N_2\}$ denote the set of all
servers and $\{m_1,\cdots,m_{i-1},m_i\}$ be the set of busy
servers. The busy servers are labelled in increasing order of the
arrival times of the jobs they serve; for example,
$m_1 \in \mathcal{M}$ is the label of the server processing the
earliest arriving job in the system. Let $n_j$ denote the number of
jobs that cannot be served by servers in the set
$\mathcal{M} \setminus \{ m_1,\cdots,m_{j-1}, m_j\}$. Under the
formulation of \cite{Visschers2012}, the state $s$ of the system is
represented as $$s=(n_i,m_i,\cdots,n_2,m_2,n_1,m_1).$$

As an illustration of this state description, for the example in
Fig. \ref{fig:cosequiv}, the state of the system is
$(3,S_3,0,S_2,0,S_1)$.  This state description is actually an aggregated
one because out of the three waiting jobs two can only be served in $S_1$
and $S_3$. For a complete state description, the type of the job
indicating the servers in which it could be served should also have
been included. However, due to the assumption of Poisson arrivals, all
the arrival streams can be aggregated into a single stream, and each
arrival can be reassigned its type at the moment when a server becomes
free. We refer the reader to \cite{Visschers2012} for details of this state space description.

Using this state description, under a certain \emph{assignment rate
	condition} on how servers are selected when an arriving job finds more
than one idle eligible server, \cite{Visschers2012} establishes a
product form stationary distribution.

To state the assignment rate condition, we need to develop the
following notation. For $\{m_1,\cdots,m_{j-2},m_{j-1}\} \subset \M,$
and $m_j \notin \{m_1,\cdots,m_{j-2},m_{j-1}\},$ let
$\lambda_{m_j}(m_1,\cdots,m_{j-2},m_{j-1})$ denote the transition rate
from the state $(n_{j-1},m_{j-1},\cdots,n_1,m_1)$ to the state
$(0,m_j,n_{j-1},m_{j-1},\cdots,n_1,m_1).$ That is, $\lambda_{m_j}$ is
the rate at which an incoming job is sent to server $m_j$ when servers
$\{m_1,\cdots,m_{j-2},m_{j-1}\}$ are busy and the others are idle. The
assignment rate condition is the following:\\
\textbf{Assignment Rate Condition:} For $1 \leq i \leq N_1 + N_2$
and for every vector $(m_1,\cdots, m_i)$ composed of elements from
$\M,$
$$\prod_{j=1}^{i}\lambda_{m_j}(m_1,\cdots,m_{j-1})=\prod_{j=1}^{i}\lambda_{m_j}(\tilde{m}_1,\cdots,\tilde{m}_{j-1})$$
for every permutation $\tilde{m}_{1},\cdots,\tilde{m}_{i}$ of
$m_1,\cdots,m_{i}.$

That it is always possible to design assignment probabilities to free
eligible servers such that the above condition is satisfied is proved
in \cite{Visschers2012}. Let us define $\lambda(\lbrace
m_1,\cdots,m_{i-1},m_{i}\rbrace)$ as the aggregate arrival rate
corresponding to jobs which cannot be served by any server in the set
$ \mathcal{M} \setminus \lbrace m_1,\cdots,m_{i}\rbrace$. The
assignment rate condition implies the following product-form
stationary distribution.
\begin{theorem}[Theorem 2 in \cite{Visschers2012}]
	\label{thm:visschers}
	Assuming the assignment rate condition, the steady state probability
	for any state $s=(n_i,m_i,\cdots,n_2,m_2,n_1,m_1)$ is given
	by:
	$$\pi(s)=\alpha_i^{n_i}\cdots\alpha_1^{n_1}\frac{\Pi_\lambda(\lbrace m_1,\cdots,m_{i-1},m_{i}\rbrace)}{i!}\pi(0)$$
	where,$$\Pi_\lambda(\lbrace m_1,m_2,\cdots,m_i\rbrace)=\prod_{j=1}^{i}\lambda_{m_j}{(\lbrace m_1,\cdots,m_{j-2},m_{j-1}\rbrace)}$$
	$$\alpha_i= \frac{1}{i} \lambda(\lbrace m_1,\cdots,m_{i-1},m_{i}\rbrace)$$
\end{theorem}

Using this stationary distribution, one can compute the waiting
probability as well as the mean response time for each partial sharing
configuration $(k_1,k_2);$ see Appendix~\ref{app:cos_metrics}. We note that while
there is some simplification achieved by specializing
Theorem~\ref{thm:visschers} to our two-class setting with three types
of servers (i.e., two dedicated pools and one common pool), the
expressions for the performance metrics, while amenable to numerical
computation, are too cumbersome for an analytical treatment of the
Pareto-frontier.

\ignore{
	\begin{lemma}
		For a system under partial sharing configuration $(k_1,k_2)$ with
		cancel-on-start, the Erlang-C probability $C_i(k_1,k_2)$ is given by:
		$$C_i(k_1,k_2)=\sum_{s \in \mathcal{P}_j} \pi(s)$$
		where,\begin{align*}
		\mathcal{P}_j = \{(n_i,m_i,\cdots,n_1,m_1):\mathcal{M}_j \subset (m_1,\cdots,m_{i-1},m_i)\}.
		\end{align*}
		and $\mathcal{M}_j$ is the set of server that can server a job of type $j$.
	\end{lemma}
	\begin{proof}
		$\mathcal{P}_j$ is the set of states where all the server that can
		serve a job of type $j$ are busy. Hence only for all states
		$s \in \mathcal{P}_j$, an arriving job of type $j$ has to wait, which
		gives us the above expression for the Erlang-C probability.
	\end{proof}
}

\subsection{Pareto-frontier of partial sharing configurations under $\cos$
	for  $N_1 = N_2 = 1$}

We now specialize to the case $N_1 = N_2 = 1,$ focusing on the
stationary waiting probability metric. To analyse the Pareto-frontier
with $\cos$ replication, we need to generalize the space of sharing
configurations $(k_1,k_2)$ to allow for real valued $k_i \in [0,1].$
We do this using randomization as follows. Provider $i$ contributes
its server to the common pool with probability $k_i;$ these actions
being independent across providers. Of course, the above probabilities
should really be interpreted as time-fractions. So the configuration
$(k_1,k_2)$ is achieved by time-sharing between the configurations
$(0,0),$ $(0,1),$ $(1,0)$ and $(1,1),$ with (long run) time-fractions
$(1-k_1)(1-k_2),$ $(1-k_1)k_2,$ $k_1(1-k_2),$ and $k_1k_2,$
respectively.

The following theorem provides a complete characterization of the
Pareto-frontier.
\begin{theorem}\label{thm:2servers}
	For $N_1=N_2=1,$ under $\cos$ replication, the Pareto-frontier is
	non-empty. Moreover, Pareto-optimal configurations for the stationary
	waiting probability metric satisfy $k_i = 1$ for some $i.$
	Specifically, the Pareto-frontier $\mathcal{P}$ is characterized as
	follows.
	\begin{enumerate} \item If $C_i(1,1) < C_i(0,0)$ $\forall\ i,$ then
		there exist uniquely defined constants $\hat{x}_1$ and $\hat{x}_2,$
		such that $\hat{x}_i \in (0,1)$ for $i = 1,2,$ $C_1(1,\hat{x}_2) =
		C_1(0,0),$ $ C_2(\hat{x}_1,1) = C_2(0,0).$ In this case,
		$$\mathcal{P} = \{(x,1):\ x \in(\hat{x}_1,1]\} \cup \{(1,x):\
		x \in(\hat{x}_2,1]\}.$$ \item If $C_2(0,0) \leq C_1(1,1) =
		C_2(1,1) < C_1(0,0),$ then there exist uniquely defined constants
		$\underline{x}_2$ and $\bar{x}_2$ satisfying $0 < \underline{x}_2
		< \bar{x}_2 \leq 1$ such that $C_1(1,\underline{x}_2) = C_1(0,0)$
		and $C_2(1,\bar{x}_2) = C_2(0,0).$ In this case,
		$$\mathcal{P} = \{(1,x):\ x \in
		(\underline{x}_2,\bar{x}_2)\}.$$
	\end{enumerate}
\end{theorem}
Theorem~\ref{thm:2servers} shows that Pareto-optimal configurations
always involve at least one of the providers always contributing its
server to the common pool. Moreover, the theorem also spells out the
exact structure of the Pareto-frontier. Case~(1) of the theorem
corresponds to the case where full pooling is beneficial to both
providers. In this case, the Pareto-frontier includes the full-pooling
configuration; see Figure~\ref{fig:n=1,case1} for an example of this
case. Case~(2) of the theorem applies to the asymmetric setting where
full pooling benefits $P_1$ but not $P_2.$ In this case, all
Pareto-optimal configurations involve the most congested provider
(i.e., $P_1$) always contributing its server to the common pool; see
Figure~\ref{fig:n=1,case2} for an example of this case. Note that the
third case where full pooling is beneficial to $P_2$ but not $P_1$ is
omitted in the theorem statement, since it can be recovered from
Case~(2) by interchanging the labels of the providers. The proof of
Theorem~\ref{thm:2servers} is provided in
Appendix~\ref{app:proof_2servers}.

We conjecture that the above structure of the Pareto-frontier holds
beyond the special case of $N_1 = N_2 = 1.$
\begin{conjecture}
	For $N_1,N_2 \geq 1,$ employing an extension of the space of partial
	sharing configurations to $[0,N_1]\times[0,N_2]$ via randomization as
	before, any Pareto-optimal configuration for the stationary waiting
	probability metric satisfies $k_i = N_i$ for some $i.$
\end{conjecture} 
Numerical experimentation suggests that the above conjecture holds;
however, a proof has eluded us thus far. See Figure~\ref{fig:larger_n}
for some illustrations of the Pareto-frontier computed for $N_1 = N_2
= 2.$
\begin{figure}[t]
	\centering
	\begin{subfigure}[t]{0.48\columnwidth}
		\centering
		\includegraphics[]{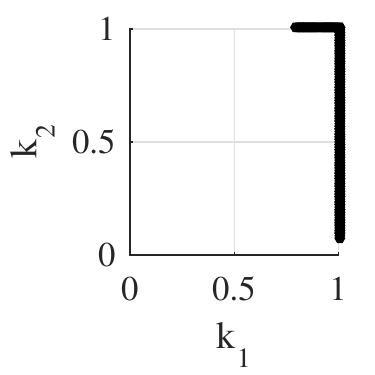}
		\caption{$C_1^s=30\%,C_2^s=10\%$}
		\label{fig:n=1,case1}
	\end{subfigure}%
	~
	\begin{subfigure}[t]{0.48\columnwidth}
		\centering
		\includegraphics[]{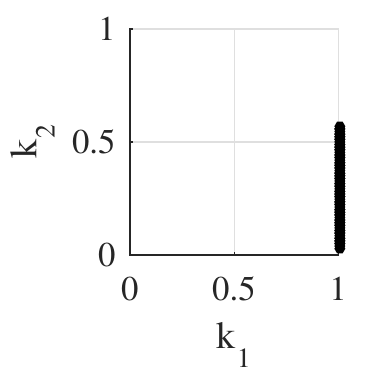}
		\caption{$C_1^s=50\%,C_2^s=10\%$}
		\label{fig:n=1,case2}
	\end{subfigure}%
	\caption{Pareto Frontier for $N_1=N_2=1$.}
	\label{fig:pareto}
\end{figure}
\begin{figure}[t]
	\centering
	\begin{subfigure}[t]{0.48\columnwidth}
		\centering
		\includegraphics[]{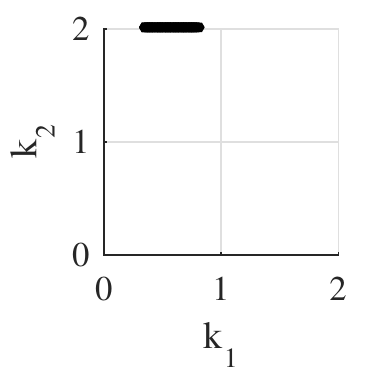}
		\caption{$C_1^s=10\%,C_2^s=50\%$}
		\label{fig:n=2,case1}
	\end{subfigure}%
	~
	\begin{subfigure}[t]{0.48\columnwidth}
		\centering
		\includegraphics[]{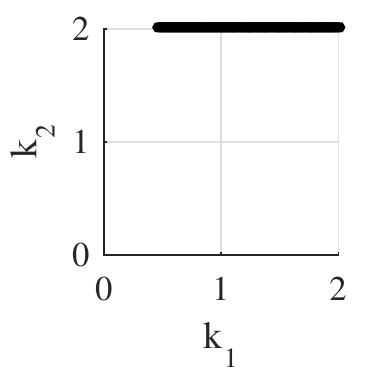}
		\caption{$C_1^s=20\%,C_2^s=50\%$}
		\label{fig:n=2,case2}
	\end{subfigure}%
	\caption{Pareto Frontier for $N_1=N_2=2$.}
	\label{fig:larger_n}
\end{figure}
\subsection{Bargaining solutions}

The Pareto-frontier defines the set of partial pooling configurations
from which one provider cannot improve its performance without
degrading that of the other. That is, one cannot find a configuration
that will be better for both providers simultaneously. Bargaining
theory provides a framework for choosing one configuration from the
set of Pareto-optimal ones \cite{Myerson2013}. While an extensive treatment of
bargaining solutions is beyond the scope of the present paper, we now
use one popular solution concept from the theory, namely, the
Kalai-Smorodinsky bargaining solution (KSBS) \cite{KS75}, to capture the
agreement point between the providers.

\begin{definition}
	A partial sharing configuration $(x_1^*,x_2^*)$ is a \emph{
		Kalai-Smorodinsky bargaining solution (KSBS)}, if
	$(x_1^*,x_2^*)$ is on the Pareto-frontier and satisfies
	\begin{align*}
	\frac{B_1(0,0) - B_1(x_1^*,x_2^*)}{ B_2(0,0) - B_2(x_1^*,x_2^*)} &~= \frac{B_1(0,0)
		- \min\limits_{y \in [0,1]^2} B_1(y_1,y_2)}{ B_2(0,0) - \min\limits_{y \in [0,1]^2} B_2(y_1,y_2)}. 
	\end{align*}
\end{definition}
The KSBS is such that the ratio of relative utilities of the providers
is equal to the ratio of their maximal relative utilities. It can be
shown (using arguments similar to those in \cite{FullVersion}) that for $N_1 =
N_2 = 1$ and for the stationary waiting probability metric, the KSBS
is uniquely defined. Table~\ref{tb:n=1} illustrates the KSBS computed for
various system parameters when $N_1 = N_2 = 1.$ We see that when the
providers are close to symmetric in their arrival rates (or
equivalently, their standalone waiting probabilities), the KSBS
corresponds to full pooling. On the other hand, when one provider is
much more congested than the other, the KSBS does not correspond to
complete pooling; only the more congested provider is required to
always place its server into the common pool. Note that the KSBS
affords a considerable improvement in performance for both providers.

\begin{table}[h]
	\centering
	\caption{Comparison of waiting probabilities (in \%) for the case $N_1 = N_2 = 1,$ respectively and $\nu_1=\nu_2=1$.}
	\label{tb:n=1}
	
	\begin{tabular}{|l|l|c|c|c|c|c|}
		\hline
		\multicolumn{2}{|c|}{Standalone}  &\multicolumn{1}{|c|}{Full Sharing }&\multicolumn{4}{|c|}{KSBS}\\
		\hline
		
		$C_1^s$&$C_2^s$ &$C_1=C_2=C$& $k_1^*$ &  $k_2^*$&$C_1(k_1^*,k_2^*)$&$C_2(k_1^*,k_2^*)$
		\\\hline
		
		10&10 &1.82 & 1&1 &1.82 &1.82  \\
		10&30 &6.65&0.69 &1 &5.54 &14.54 \\
		10&50 & 13.85 & 0.37&1&8.26 &37.30    \\
		\hline
	\end{tabular}
\end{table}

\section{Discussions}
\label{section:discussions}

In this section, we cover the special (and easier) single-server
setting, where the service capacities of both providers can be merged
into a single server. We also review the related literature, and
outline potential directions for future work.

\subsection{Single server setting}
First, we consider the special case where the service capacity of both
providers can be combined into a single server. Moreover, the service
capacity of this merged server is arbitrarily (and dynamically)
divisible between the two providers. For this case, we define a
balanced fairness (BF) based partial pooling mechanism between the
providers, which includes both no pooling and complete pooling as
special cases. For this mechanism, we show that complete pooling is
the only Pareto-optimal configuration. Of course, given the
equivalence between $\coc$ replication and single-server balanced-fair
scheduling, and our conclusions from Section~\ref{section:coc}, this
is not surprising.

Consider two service providers $P_1$ and $P_2$ with servers operating
at the rate (a.k.a. speed) $\mu_1$ and $\mu_2$ respectively. Recall
that in this section, we consider the `single server' setting, where
it is possible to pool the capacity of the two servers into a single
server with rate $\mu_1 + \mu_2.$ Jobs of $P_i$ arrive according to a
Poisson process of rate $\lambda_i$, and the service requirements
(a.k.a. sizes) of jobs are i.i.d. and exponentially distributed with
mean $1/\nu_i < \infty.$ Let $\rho_i := \lambda_i/\nu_i$ denote the
traffic load corresponding to $P_i.$ For stability, we assume that
$\rho_i < \mu_i$ for $i = 1,2.$

Jobs corresponding to each provider wait in separate queues, and are
served in a FCFS fashion. The service rate for each queue is
determined by a BF-based partial pooling mechanism, parametrized by
$(k_1,k_2).$ Here, $k_i \in [0, \mu_i]$ is a measure of the extent to
which $P_i$ is willing to share its service capacity with $P_{-i}.$
Specifically, let $n_i$ denote the number of unfinished jobs of $S_i$
in the system. Given the system state $(n_1,n_2),$ the BF-based
partial pooling mechanism awards a service rate $r_i$ to Provider~$i,$
such that $(r_1,r_2)$ is the balanced fair allocation over the rate
region $\mathcal{R}(k_1,k_2)$ defined as follows:
\begin{align*}
\mathcal{R}(k_1,k_2) = \{(r_1,r_2) \in\, &\R_+^2:\ r_1\leq \mu_1+k_2,\;
r_2\leq \mu_2+k_1,\;\\
& \quad r_1+r_2 \leq \mu_1+\mu_2 \}.
\end{align*}
\TRtext{The shape of above rate region is similar to the rate region in 
	Figure \ref{fig:bf-rateregion} with $N_i$ replace by $\mu_i$.}
\begin{theorem}\label{thm:single_server_pareto}
	In the single server setting, under the BF-based partial pooling
	mechanism, the full sharing configuration (i.e., $k_1=\mu_1,
	k_2=\mu_2$) is the only Pareto-optimal configuration, for the
	stationary waiting probability metric as well as the stationary mean
	response time metric.
\end{theorem}
Theorem~\ref{thm:single_server_pareto} implies that the resource
pooling problem is easy in the single server setting---complete
pooling with balanced fair resource allocation is the solution.



\subsection{Related work}

Complete resource pooling between independent service systems has been
studied from a cooperative game theory standpoint; see, for
example,~\cite{Gonzalez2004,Anily10,Karsten11b}. The goal of this
literature is to analyze stable mechanisms for sharing the surplus
(and costs) of the grand coalition among the various agents. Single
server as well as multi-server settings have been considered, for
queueing as well as loss systems; see \cite{Karsten11b} for a
comprehensive survey of this literature. A complementary view of
resource sharing comes from an optimization standpoint. Here the
organization is interested in optimally provisioning (potentially
heterogenous) resources and/or sharing its service resources between
various activities; see, for example,
\cite{Benjaafar1995,Wallace2005,Iravani2007}.

In contrast, our approach in this work is to analyze resource sharing
between strategic service providers having non-transferable
utility. In other words, no side-payments are allowed between the
service providers. Instead, providers would (partially or completely)
pool their resources with one another only if their service quality
improves in the process. Our goal is thus to devise mechanisms that
guarantee mutually beneficial sharing configurations. The only prior
work we are aware of that takes this view is \cite{FullVersion}, which
considers loss systems. In contrast, the present paper considers
service systems where jobs can be queued. As it turns out, the
appropriate sharing mechanisms are very different between these two
settings.

\subsection{Naive partial sharing}
In this paper, we analyzed partial sharing for systems with full
redundancy, i.e., copies are sent to all the servers. An alternative
system that would have been a natural generalization of the partial
sharing with blocking analyzed in \cite{FullVersion} would be to have
three separate queues---one each for the two dedicated servers and one
for the shared servers---and send an incoming job with priority to the
dedicated servers. Repacking could have been used to dynamically
exchange servers between the pools in order to ensure that the
stability region is the same as in this paper. Compared to the
redundancy system, the naive system presents an inconvenience in that
it is not easy to show whether or not it admits a product-form. For
the computation of the waiting probability, this is a major
disadvantage. One could in theory use sample-path techniques employed
in \cite{FullVersion} to characterize the Pareto-frontier, but these
become much more involved for the system in this paper and we are not
sure that they are even true. This is part of our ongoing work. 

\subsection{Future work}
An immediate avenue for future work is of course to prove Conjecture~1
and completely analyse the Pareto-frontier under $\cos$ replication;
this is currently being pursued. It would also be interesting to
generalize our models to $n > 2$ service providers, and to devise
suitable partial sharing mechanisms that guarantee the existence of
mutually beneficial sharing configurations.

More broadly, the present bargaining-centric view of resource pooling
(with non-transferable utilities between agents) can be explored in
other contexts; for example, sharing of cache memory between content
distribution systems, sharing of energy storage in smart power grids,
and spectrum sharing between cellular service providers or between
secondary users in cognitive radio networks.

\bibliographystyle{IEEEtrans}
\bibliography{main} 
\appendices

\TRtext{
	
	\section{Proof of Lemma ~\ref{lemma:coc_waiting}}
	\label{app:coc_waiting}
	We derive the expression for mean response time for a partially pooled system with sharing configuration $(k_1,k_2)$ with c.o.c. replication using the equivalence of its steady state system occupancy vector with that of balanced fair scheduling as mentioned in Section \ref{section:coc}. Theorem 4 in \cite{Virag2014} provides following general expression mean response time under balanced fair scheduling and thus also for for a partially pooled system with sharing configuration $(k_1,k_2)$ with c.o.c. replication.
	\begin{equation}\label{eq:mean_delay}
	D_i(k_1,k_2)=\frac{\nu_i \frac{\partial G(k_1,k_2)}{\partial \rho_i}}{G(k_1,k_2)}.
	\end{equation} 
	Given $G(k_1,k_2)$, by differentiating it wrt  $\rho_i$ and substituting it in Equation \ref{eq:mean_delay}, we get the expression for mean response time.
	
	$G(k_1,k_2)$ can be calculated using following equations stated in \cite{Virag2014}:
	Let $G_{\lbrace \phi\rbrace}=1$ and for $\mathcal{A} \in \{\{1\},\{2\},\{1,2\}\}$, $G_{\mathcal{A}}$ is defined as:
	$$G_{\mathcal{A}}=\frac{\sum_{i \in A}\rho_iG_{A-i}}{\mu(\mathcal{A})-\sum_{j \in \mathcal{A}}\rho_j}.$$
	$\mu(\mathcal{A})$ is the maximum service rate for the set of providers $\mathcal{A}$ in the balanced fair rate allocation as shown in Figure \ref{fig:bf-rateregion}. Hence,
	$$G_{\{ 1\}}=\frac{\rho_1}{N_1+k_2-\rho_1},
	G_{\{ 2\}}=\frac{\rho_2}{N_2+k_1-\rho_2},$$
	$$G_{\{ 1,2\}}=\frac{\rho_1 G_{\{ 2\}}+\rho_2G_{\{ 1\}}}{N_1+N_2-\rho_1-\rho_2}=\rho_1\rho_2\frac{\frac{1}{N_1-\rho_1}+{\frac{1}{N_2-\rho_2}}}{N_1+N_2-\rho_1-\rho_2}.$$\\
	$G(k_1,k_2)=G_{\lbrace \phi\rbrace}+G_{\lbrace1\rbrace}+G_{\lbrace2\rbrace}+G_{\lbrace1,2\rbrace}$. By substituting the values $G_{\{\phi\}},G_{\{1\}}G_{\{2\}}$, $G_{\{1,2\}}$ and rearranging some terms, we get
	\begin{align*}
	G(k_1,k_2)=\frac{1}{1-\frac{\rho_1+\rho_2}{N_1+N_2}}\bigg(\frac{1-\frac{\rho_1}{N_1+k_2}}{1-\frac{\rho_1}{N_1+N_2}}+\frac{1-\frac{\rho_2}{N_2+k_1}}{1-\frac{\rho_2}{N_1+N_2}}-1\bigg).
	\end{align*}
} 

\TRtext{
	\section{Proof of Lemma ~\ref{lemma:coc_mono}}
	\label{appendix:coc_mono}
	1)
	The mean response time of the overall system $D(k_1,k_2)$ is given by:
	$D(k_1,k_2)=\frac{\lambda_1}{\lambda_1+\lambda_2}D_1(k_1,k_2)+\frac{\lambda_2}{\lambda_1+\lambda_2}D_2(k_1,k_2)$. On simplification we get,
	$$D(k_1,k_2)=\frac{\rho_1+\rho_2}{N_1+N_2-\rho_1-\rho_2)}+\frac{L_1(k_1,k_2)+L_2(k_1,k_2)}{H(k_1,k_2)}$$
	where,\\
	$H(k_1,k_2)=\bigg(\frac{1-\frac{\rho_1}{N_1+N_2}}{1-\frac{\rho_1}{N_1+k_2}}+\frac{1-\frac{\rho_2}{N_1+N_2}}{1-\frac{\rho_2}{N_2+k_1}}-1\bigg)$ and  $L_i(k_1,k_2)=\frac{\frac{p_i}{N_i+k_{-i}}-\frac{p_i}{N_1+N_2}}{\big(1-\frac{p_i}{N_i+k_{-i}}\big)^2}.$ Now we show that the overall mean response time decrease with $k_i$. 
	\begin{align*}
	\frac{\partial{D}}{\partial k_{2}}=&\frac{-\rho_1}{H^2(N_i+k_{-i})^2}\Bigg[\frac{\big(1-\frac{\rho_1}{N_1+N_2}\big)^2}{\big(1-\frac{p_1}{N_1+k_2}\big)^4}\\&+\frac{L_2\big(1-\frac{\rho_2}{N_2+k_1}\big)}{\big(1-\frac{p_1}{N_1+k_2}\big)^2}\Bigg(2\frac{\big(1-\frac{\rho_1}{N_1+N_2}\big)}{\big(1-\frac{p_1}{N_1+k_2}\big)}-1\Bigg)+L_2\Bigg]<0.
	\end{align*}
	$\frac{\partial{D}_1}{\partial k_{2}}=\frac{\lambda_1+\lambda_2}{\lambda_1}\Big(\frac{\partial{D}(k_1,k_2)}{\partial k_{2}}-\frac{\lambda_1}{\lambda_1+\lambda_2}\frac{\partial{D}_1(k_1,k_2)}{\partial k_{2}}\Big) <0.$\\\\
	2) If $k_{-i}=N_{-i}$,then is easy to note that $D_i(k_1,k_2)$ depends on $k_i$ only through $H(k_1,k_2)$. Since $H(k_1,k_2)$ is strictly decreasing function of $k_i$, $D_i(k_1,k_2)$ increases with $k_i$.\\\\
	3) If $k_{-i}=N_{-i}$,
	$D_i(k_1,k_2)=\frac{\nu_i}{N_1+N_2-\rho_1-\rho_2}$, does not depend on $k_i$.
	
}

\section{Performance Characterization with $\cos$}
\label{app:cos_metrics}
\subsection{Assignment Rates}
We provide a detailed approach for computing the steady-state distribution and hence the expressions for $C_i(k_1,k_2)$ and $D_i(k_1,k_2)$. We introduce the following notations:
\begin{align*}
\mathcal{M}_1=&\{1,\cdots N_1-k_1\}\text{, dedicated servers of } P_1,\\
\mathcal{M}_2=&\{N_1+k_2+1,\cdots N_1+N_2\}\text{, dedicated servers of } P_2,\\
\mathcal{M}_3=&\{N_1-k_1+1,\cdots N_1+k_2\}\text{, shared servers.}
\end{align*}
For a given subset of servers $m=\{m_1,\cdots,m_j,m_k\}$, define
$x_i(m) =$  number of servers from the set $\mathcal{M}_i$ in $m$ and $x(m)=(x_1(m),x_2(m),x_3(m))$. Also define $\bar{x}(m)=(\bar{x}_1(m),\bar{x}_2(m),\bar{x}_3(m))$
where $\bar{x}_i(m)=|\mathcal{M}_i|-x_i(m)$ and $e_i$ be the unit vector in the direction $i$.

All servers in $\mathcal{M}_i$ are identical, hence $\lambda_{m_j}\{m\}$ is same for all $m_j \in \mathcal{M}_i$ for $i=1,2,3.$ Using the assignment rate condition it can also be shown that  $\lambda_{m_j}\{m\}=\lambda_{m_j}\{m'\}$ if $x(m)=x(m')$. Hence, we can define $r_i(x(m))=\lambda_{m_j}\{m\}$ for $m_j \in \mathcal{M}_i$ and $m \in \mathcal{M}$.
Since the vector $x$ takes values from a finite set  $\mathcal{X} ( \in \mathds{Z}^3$), calculating $r_i(x)$ for all values in $\mathcal{X}$ gives all the assignment rate probabilities $\lambda_{m_j}\{m\}$. 

Recall that $\lambda(m)$ is the aggregate arrival rate
corresponding to jobs which cannot be served by any server in the set $\mathcal{M} \setminus m$. Define $\bar{\lambda}(m)= \lambda_1+\lambda_2 -\mathcal{M} \setminus m$, the aggregate arrival rate
corresponding to jobs which can be served by the servers in $\mathcal{M} \setminus m$, note that $\bar{\lambda}(m)=\bar{\lambda}(m')$ if $x(m)=x(m')$. Hence $r(x(m))=\bar{\lambda}(m)$. Using this and the assignment rate condition  stated in Section \ref{section:cos}, we obtain the following set of equations between $r_i(x)$, which can be solved to get $r_i(x)$ for all $x \in \mathcal{X}$ and $i=1,2,3$.
$$\bar{x}_1r_1(x)+\bar{x}_2r_2(x)+\bar{x}_3r_3(x)=r(x).$$
For $  x \in \mathcal{X}\text{ such that } x_1<N_1-k_1 \text { and } x_3<k_1+k_2,$
$$r_1(x)r_3(x+e_1)=r_1(x+e_3)r_3(x).$$
For $  x \in \mathcal{X}\text{ such that } x_2<N_2-k_2 \text { and } x_3<k_1+k_2,$
$$r_2(x)r_3(x+e_2)=r_2(x+e_3)r_3(x).$$ 
$$r_i((N_1-k_1,N_2-k_2,k_1+k_2)-e_i)=\begin{cases}
\lambda_1 \text{ if } i=1,\\
\lambda_2 \text{ if } i=2,\\
\lambda_1+\lambda_2 \text{ if } i=3.
\end{cases}$$
$$r_i(x)=0 \text{ if } x_i=|\mathcal{M}_i|.$$
$$\text{Recall that, } \Pi_\lambda(\lbrace m_1,\cdots,m_i\rbrace)=\prod_{j=1}^{i}\lambda_{m_j}{(\lbrace m_1,\cdots,,m_{j-1}\rbrace)}.$$
Again due to the assignment rate condition, $\Pi_\lambda(m)=\Pi_{\lambda}(m')$ if $x(m)=x(m')$, hence we define $\Pi(x(m))=\Pi_\lambda(m)$.
\subsection{The Waiting Probability}
For a state  $s=(n_i,m_i,\cdots,n_2,m_2,n_1,m_1)$, $m(s)$ is the set of busy servers., i.e, $m(s)=\{m_1,m_2\cdots, m_i\}$. then the waiting probability $C_i(k_1,k_2)$ is follows:
$$C_1(k_1,k_2)=\sum_{s:x_1(m(s))+x_3(m(s))=N_1+k_2.} \pi(s).$$
Let us define $\pi_{\{\mathcal{A}\}}$ as the probability that an arriving job of the provider $\mathcal{A}$ has to wait ($\mathcal{A} \in\{\{\phi\},\{1\},\{2\},\{1,2\}\}$). Equivalently, $\pi_{\{\mathcal{A}\}}$ is the sum of stationary probabilities of those states in which all servers which can serve jobs from providers $\mathcal{A}$ are busy. Then the waiting probabilities for Provider $P_i$ can be expressed as, $$C_i(k_1,k_2)=\pi_{\{i\}}+\pi_{\{1,2\}},$$
where,
$\pi_{\{1\}}=\sum_{s:x_1(m(s))+x_3(m(s))=N_1+k_2,x_2(m(s)<N_2-k_2}\pi(s).$
After obtaining the assignment probabilities and then using the expression for $\pi(s)$ from Theorem \ref{thm:visschers}, we get  $\pi_{\{1\}}=\sum\limits_{x_2=0}^{N_2-k_2-1}\frac{\Pi(N_1-k_1,y,k_1+k_2)\pi(0)}{(N_1+k_2+x_2)!}{N_2-k_2 \choose y}\Bigg(\sum_{w=0}^{x_2}{x_2 \choose w}w!(N_1+k_2)(N_1+k_2+x_2-w-1)!\bigg(\Pi_{v=0}^w\frac{(N_1+k_2+x_2-w+v)}{(N_1+k_2+x_2-w+v-\lambda_1)}\bigg)\Bigg).$
Similarly, we can write the expression for $\pi_{\{2\}}$ using symmetry, also
$\pi_{\{1,2\}}=\sum_{s:x_1(m(s))+x_2(m(s))+x_3(m(s))=N_1+N_2}\pi(s)$,\\\\
$\pi_{\{1,2\}}=\frac{\Pi(N_1-k_1,N_2-k_2,k_1+k_2)\pi(0)}{(N_1+N_2-1)!(N_1+N_2-\lambda_1-\lambda_2)}\Bigg(\sum_{w=0}^{N_2-k_2}{N_2-k_2 \choose w}w!\\(N_1+k_2)(N_1+N_2-w-1)!\bigg(\Pi_{v=0}^{w-1}\frac{(N_1+N_2-w+v)}{(N_1+N_2-w+v-\lambda_1)}\bigg)\Bigg)$.

Similarly the expression for mean response time can be computed, which is omitted due to lack of space. As one can see these expressions though easy to compute numerically are not very amenable to analysis.
\section{Proof of Theorem ~\ref{thm:2servers}}
\label{app:proof_2servers}
The proof of Theorem \ref{thm:2servers} is structured as follows:
first we show that $C_i(k_1,k_2)$ satisfies a set of 3 conditions in Lemma \ref{lemma:cos_mono}. Then we show that if those conditions are satisfied then Pareto region lies on the boundary, i.e. at least one of the $k_i=N_i$. After showing that the Pareto-frontier lies on the boundary, the explicit structure of the Pareto-frontier described in Theorem \ref{thm:2servers} is direct consequence of the monotonicity properties as shown in \cite{FullVersion}. The proof is omitted due to lack of space. 
\begin{lemma}\label{lemma:cos_mono}
	For $\cos$ replication and $N_1=1$ and $N_2=1$	
	\begin{enumerate}
		\item $C_i(k_1,k_2)$ is a strictly decreasing function of $k_{-i},$
		\item $C_i(k_1,k_2)$ is a strictly increasing function of $k_{i},$
		\item $\frac{\partial C_1}{\partial k_2}\frac{\partial C_2}{\partial k_1}>\frac{\partial C_1}{\partial k_1}\frac{\partial C_2}{\partial k_2}.$
	\end{enumerate}
\end{lemma}
\begin{proof}
	As defined in Section \ref{section:cos}, for any $k_1 \in [0,1] \text{ and } k_2 \in [0,1],$ the waiting probability is given by:
	$$C_i(k_1,k_2)=(1-k_1)(1-k_2)C_i(0,0)+(k_1)(1-k_2)C_i(1,0)
	$$$$+(1-k_1)k_2C_i(0,1)+k_1k_2C_i(1,1),$$
	where $C_i(0,0),C_i(1,0),C_i(0,1)$ and $C_i(0,1)$ can be explicitly computed using the results in Appendix \ref{app:cos_metrics}. For the sake of tidiness, we use the notation $C_i$ instead of $C_i(k_1,k_2)$ for the rest of the proof.
	For $i=1$,
	\begin{align*}
	\frac{\partial C_1}{\partial k_1}
	=&(1-k_2)\big[C_1(1,0)-C_1(0,0)\big]+k_2[C_1(1,1)-C_1(0,1)\big]\\
	=&\frac{(1-k_2)}{\Omega_1}(2-\lambda_1+\lambda_2-\lambda_1\lambda_2-\lambda_1^2)\lambda_2^2\\
	&+\frac{k_2}{\Omega_2}(2-\lambda_1-\lambda_2)\lambda_2(\lambda_1+\lambda_2)^2>0.\\
	\frac{\partial C_1}{\partial k_2}
	=&(1-k_1)\big[C_1(0,1)-C_1(0,0)\big]
	+k_1[C_1(1,1)-C_1(1,0)\big]\\
	=&\frac{(1-k_1)}{\Omega_2}\lambda_1(1-\lambda_2)[\lambda_1(\lambda_2-4)+\lambda_2(\lambda_2-2)]\\
	&+\frac{k_1}{\Omega_1}(\lambda_1\lambda_2+\lambda_2^2+2\lambda_1-4)(\lambda_1+\lambda_2)^2<0,
	\end{align*}
	where\\
	$\Omega_1
	=(\lambda_1+\lambda_2+\lambda_2^2)(1-\lambda_1)+\lambda_1(1-\lambda_1\lambda_2)+3\lambda_2+\lambda_2^2
	>0,$\\
	$\Omega_2=(\lambda_1+\lambda_2+\lambda_1^2)(1-\lambda_2)+\lambda_2(1-\lambda_1\lambda_2)+3\lambda_1+\lambda_1^2
	>0.$
	Similarly due to symmetry, it can be shown that $\frac{\partial C_2}{\partial k_1}< 0$ and $\frac{\partial C_2}{\partial k_2}> 0$. Hence, $C_i(k_1,k_2)$ is a strictly decreasing function of $k_{-i},$ and $C_i(k_1,k_2)$ is a strictly increasing function of $k_{i}.$ 
	$$\frac{\partial C_1}{\partial k_2}\frac{\partial C_2}{\partial k_1} - \frac{\partial C_1}{\partial k_1}\frac{\partial C_2}{\partial k_2}\nonumber
	=\; \alpha k_1+\beta k_2+\gamma,$$
	where,
	\begin{align*}
	\alpha=&\big[C_1(1,0)-C_1(0,0)\big]\big[C_2(0,1)-C_2(1,1)\big]-\\&\big[C_1(1,1)-C_1(0,1)\big]\big[C_2(0,0)-C_2(1,0)\big]\\
	=&\frac{2(1-\lambda_1)(2-\lambda_1-\lambda_2)(2+\lambda_1-\lambda_2+\lambda_1^2+\lambda_1\lambda_2)}{(2+\lambda_1+\lambda_2)\Omega_1\Omega_2}>0,\\
	\beta=&\big[C_1(1,0)-C_1(1,1)\big]\big[C_2(0,1)-C_2(0,0)\big]-\\&\big[C_1(0,0)-C_1(0,1)\big]\big[C_2(1,1)-C_2(1,0)\big]\\
	=&\frac{2(1-\lambda_2)(2-\lambda_1-\lambda_2)(2+\lambda_2-\lambda_1+\lambda_2^2+\lambda_1\lambda_2)}{(2+\lambda_1+\lambda_2)\Omega_1\Omega_2}>0,\\
	\gamma=&\big[C_1(0,0)-C_1(0,1)\big]\big[C_2(0,0)-C_2(1,0)\big]-\\&\big[C_1(1,0)-C_1(0,0)\big]\big[C_2(0,1)-C_2(0,0)\big]\\
	=&\frac{4\lambda_1\lambda_2(1-\lambda_1)(1-\lambda_2)(2-\lambda_1-\lambda_2)}{\Omega_1\Omega_2}>0.
	\end{align*}
	Hence we have,
	$\frac{\partial C_1}{\partial k_2}\frac{\partial C_2}{\partial k_1} - \frac{\partial C_1}{\partial k_1}\frac{\partial C_2}{\partial k_2}>0.$
\end{proof}
Now given that the Lemma \ref{lemma:cos_mono} holds, we show that the
Pareto Frontier lies on the boundary as follows: we show that for any
$(k_1,k_2)$ where $k_i \in [0,N_i)$, there exists a $\theta$ such that
$\nabla C_i(k_1,k_2)\cdot(1,\theta)< 0 \text{ for } i=1,2.$
$C_1(k_1,k_2)\cdot(1,\theta)> 0$ and $C_2(k_1,k_2)\cdot(1,\theta)> 0$
are equivalent, respectively, to 
$$\theta>-\frac{\frac{\partial C_1}{\partial
		k_1}}{\frac{\partial C_1}{\partial k_2}}, \quad
\theta<-\frac{\frac{\partial C_2}{\partial k_1}}{\frac{\partial C_2}{\partial k_2}}.$$
Therefore if Statement~3 of Lemma~\ref{lemma:cos_mono} holds, there exists a direction $\theta$ along which both providers will have lower values of $C_i(k_1,k_2)$ and hence any such
configuration cannot be Pareto-optimal.
\clearpage
\end{document}